\documentclass[11pt]{article}
\usepackage{amsfonts}
\usepackage{mathrsfs}

\usepackage[latin1]{inputenc}
\usepackage{float}
\usepackage{amsmath,amssymb}
\usepackage{latexsym}
\usepackage{epstopdf}
\usepackage{geometry}  
\usepackage[active]{srcltx}
\usepackage{graphicx}
\usepackage{epstopdf}
\usepackage{booktabs}
\usepackage{multirow}
\usepackage{siunitx}

\usepackage[
bookmarks=true,         
bookmarksnumbered=true, 
colorlinks=true, pdfstartview=FitV, linkcolor=blue, citecolor=blue,
urlcolor=blue]{hyperref}

\newtheorem {theorem}{Theorem}[section]

\newtheorem{definition}{Definition}[section]

\newtheorem{lemma}{Lemma}[section]

\newtheorem{remark}{Remark}[section]

\newenvironment{proof}[1][Proof]{\textbf{#1.} }{\
\rule{0.5em}{0.5em}}

\newcommand{\bi}[1]{\mbox{\boldmath{$ #1 $}}}

\def\E{{{\mathbb E}\,}}

\begin{document}

\title{Empirical Likelihood Test for Diagonal Symmetry}
\author{Yongli Sang\textsuperscript{a}\thanks{CONTACT Yongli Sang. Email: yongli.sang@louisiana.edu} and Xin Dang\textsuperscript{b}}
\date{%
\textsuperscript{a}Department of Mathematics, University of Louisiana at Lafayette, Lafayette, LA 70504, USA\\
    \textsuperscript{b}Department of Mathematics, University of Mississippi, University, MS 38677, USA\\[2ex]%
         \today
}

\maketitle

\begin{abstract}
Energy distance is a statistical distance between the distributions of random variables, which characterizes the equality of the distributions. 
Utilizing the energy distance, we develop a nonparametric test for the diagonal symmetry, which is consistent against any fixed alternatives.  The test statistic developed in this paper is based on the difference of two $U$-statistics. By applying the jackknife empirical likelihood approach, the standard limiting chi-square distribution with degree freedom of one is established and is used to determine critical value and $p$-value of the test. Simulation studies show that our method is competitive in terms of empirical sizes and empirical powers.   
\noindent  

\vskip.2cm 

\noindent {\bf Keywords:}
\noindent  Diagonal symmetry, Energy distance, Jackknife empirical likelihood,  $U$-statistic, Wilks' theorem


\end{abstract}

\section{Introduction}
\noindent
Testing the departures from symmetry is crucial in statistics, especially in nonparametric statistical science. This problem has been discussed by many works for the univariate case (see e.g. \cite{Antille1982, Bhattacharya1982, Boos1982, Cabilio1996, Doksum1977, Hollander1988, Kou1985, Koziol1985, Miao2006, Mira1999, Randles1980, Schuster1987}).
For the multivariate case, notions of symmetry include spherical, elliptical, central, and angular symmetry.  One can refer to \cite{Serfling2006} for perspectives on those symmetries. The central symmetry is considered in this paper (see e.g. \cite{EG16, Ghosh1992, Heathcote1995, Henze2003, Neuhaus1998, Szekely2001}). 
In this setting, we explore testing the central symmetry of a distribution around the specified center. Without loss of generality, the center at $\bi 0$ is assumed and the central symmetry at $\bi 0$ is called the diagonal symmetry. Suppose $\bi X$ is a $d$-variate random variable from an unknown distribution $F$ in $\mathbb R^d$, then the hypothesis of interest  is 
\begin{align}\label{test:symmetry}
H_0: \bi X \stackrel{D}{=} -\bi X,
\end{align}
where  $\stackrel{D}{=}$ denotes equality in distribution. 

Often an omnibus test is based on a discrepancy measure among distributions. The energy distance (\cite{Szekely13, Szekely17}) is one of such measures, which is defined as follows. 
\begin{definition}[Energy distance]
Suppose that $(\bi X, \bi X') $ and $(\bi Y, \bi Y') $ are independent pairs independently from d-variate distributions $F$ and $G$, respectively. Then the energy distance between $\bi X$ and $\bi Y$ is  
\begin{align} \label{energy-distance}
\mathcal{E}(\bi X,\bi Y)=2\E\|\bi X-\bi Y\|-\E\|\bi X-\bi X'\|-\E\|\bi Y-\bi Y'\|,
\end{align}
where $\|\cdot\|$ denotes the Euclidean distance.
\end{definition}
 Sz{\'e}kely and M{\'o}ri (\cite{Szekely2001}) have shown  that the central symmetry at $\bf 0$  holds if and only if $\mathbb{E}\|\bi X+\bi X'\| = \mathbb{E}\|\bi X-\bi X'\|$. 
Utilizing this fact, we propose an empirical likelihood test for the diagonal symmetry.  


Empirical likelihood (EL) tests (\cite{EM03, Cao06, Zhang2007}) have been proposed to avoid the time-consuming permutation procedure. The EL approach  (\cite{Owen1988}, \cite{Owen1990}) enjoys effectiveness of likelihood method and flexibility of nonparametric method,  and hence has been widely used.  On the computational side, the EL method involves a maximization of the nonparametric likelihood supported on data subject to some constraints. If these constraints are linear, the computation of the EL method is relatively easy. However,  EL loses this computational efficiency when some nonlinear constraints are involved. Jackknife empirical likelihood (JEL) (\cite{Jing2009}) is proposed to reduce the computation burden of the  EL approach when nonlinear constraints are involved.  The JEL has been proved very efficient in dealing with $U$-statistics and has attracted statistician's interest due to the efficiency. In the literature, there are a quantity of papers utilizing the JEL to conduct tests, see \cite{Sang2019a}, \cite{Sang2019b},\cite{Liu2018}, \cite{Wan2018} and so on.

As mentioned above, suppose $\bi X$ and $\bi X'$ are i.i.d. from distribution $F$  in $\mathbb{R}^d$, then 
testing the diagonal symmetry of (\ref{test:symmetry}) will be equivalent to testing
\begin{align}\label{test1}
H'_0: \mathbb{E}\big [\|\bi X+\bi X'\|-\|\bi X-\bi X'\| \big ]=0.
\end{align}
The left-side of the  equation is half of the energy distance between $\bi X$ and $-\bi X^\prime$ and can be naturally estimated by the energy statistic that is the difference of two $U$-statistics.   However, the energy statistic is degenerate under $H_0$. This does not allow us to apply the regular JEL of \cite{Jing2009} directly. We develop an improved JEL for such a degenerate $U$-statistic case in this paper,  which enables us to test the diagonal symmetry efficiently. 

The reminder of the paper is organized as follows. In Section 2, we develop the JEL method for the diagonal symmetry  test. Simulation studies are conducted in Section 3. Section 4  concludes the paper with a brief summary. All proofs are reserved to the Appendix. 

\section{Methodology}

Let $\{\bi X_i\}^n_{i=1}$ be a random sample from $F$, and define the energy statistic, 
\begin{align*}
U={n \choose 2}^{-1} \sum_{1\leq i <j \leq n} \Big [ \|\bi X_i+\bi X_j\|-\|\bi X_i-\bi X_j\|\Big ],
\end{align*}
which is a one-sample $U$-statistic of degree 2. However, it is degenerate under $H_0$ which does not allow us to apply the JEL directly.
We develop an improved JEL approach for degenerate $U$-statistics in this section.

We partition the data set $\{\bi X_i\}^n_{i=1}$ into two parts: $\mathcal{ X}=\{\bi X_1, ..., \bi X_{n_1}\}$
and $\mathcal{Y}=\{\bi Y_1, ..., \bi Y_{n_2}\}$, where $\bi Y_i=\bi X_{n_1+i}$, $i=1,...,n_2$ and $n_1+n_2=n$.

Define 
\begin{align*}
&U_1={n_1 \choose 2}^{-1} \sum_{1\leq k <j \leq n_1}  \|\bi X_k+\bi X_j\|:={n_1 \choose 2}^{-1} \sum_{1\leq k <j \leq n_1}  h_1(\bi X_k, \bi X_j)\\
&U_2={n_2 \choose 2}^{-1} \sum_{1\leq k <j \leq n_2}  \|\bi Y_k-\bi Y_j\|:={n_2 \choose 2}^{-1} \sum_{1\leq k <j \leq n_2}  h_2(\bi Y_k,  \bi Y_j),
\end{align*}
where $h_1(\bi s, \bi t)= \|\bi s+\bi t\|$ and $h_2(\bi s, \bi t)= \|\bi s-\bi t\|$.  $U_1$ and $U_2$ are unbiased estimators for $\mathbb{E}\|\bi X+\bi X'\|$ and $\mathbb{E}\|\bi X-\bi X'\|$, respectively.
Thus, by (\ref{test1}), testing $H_0$ against any alternative hypothesis will be equivalent to testing $\mathbb{E} U_1=\mathbb{E} U_2$.

In order to develop the JEL approach, the corresponding jackknife pseudo-value for $U_1$ and $U_2$ are respectively given by 
\begin{align*}
&\hat{V}^{(1)}_{i}=n_1U_1-(n_1-1)U^{(-i)}_{1, n_1-1}, \ i=1,..., n_1,\\
&\hat{V}^{(2)}_{i}=n_2U_2-(n_2-1)U^{(-i)}_{2, n_2-1}\ i=1,..., n_2,
\end{align*} where 
\begin{align*}
&U^{(-i)}_{1, n_1-1}={n_1-1 \choose 2}^{-1}\sum_{1\leq j\leq k\leq n_1, j, k \neq i}h_1(\bi X_j, \bi X_k),\\
&U^{(-i)}_{2, n_2-1}={n_2-1 \choose 2}^{-1}\sum_{1\leq j\leq k \leq n_2, j, k \neq i}h_2(\bi Y_j, \bi Y_k).
\end{align*}
Those jackknife pseudo values are asymptotically independent (\cite{Jing2009}, \cite{Shi1984})
and 
 \begin{align*}
&U_{1}=\frac{1}{n_1}\sum_{i=1}^{n_1}\hat{V}^{(1)}_{i},\\
&U_{2}=\frac{1}{n_2}\sum_{i=1}^{n_2}\hat{V}^{(2)}_{i}.
\end{align*}
Under $H_0$, $\E(V^{(1)}_i)=\E(V^{(2)}_i):=\theta_0,$ where $\theta_0=\E \|\bi X_1-\bi X_2\|$.

To apply JEL to these jackknife pseudo values, 
let $\bi p=(p_1,...,p_{n_1})^T$ and $\bi q=(q_1,...,q_{n_2})^T$ be probability vectors. Then the maximum constricted jackknife empirical likelihood is 
\begin{align*} 
L=\max_{\bi p, \bi q, \theta} \left \{\left (\prod_{i=1}^{n_1}p_i\right )  \left ( \prod_{j=1}^{n_2}q_j \right )\right \},
\end{align*}
subject to the following constraints
\begin{align*}
&p_i \ge 0,  \  i=1,...,n_1; \ \sum_{i=1}^{n_1} p_i=1; \ \sum_{i=1}^{n_1}p_i (\hat{V}^{(1)}_i-\theta)=0,\\
&q_j \ge 0,  \  j=1,...,n_2; \ \sum_{j=1}^{n_2} q_j=1; \ \sum_{j=1}^{n_2}q_j (\hat{V}^{(2)}_j-\theta)=0.
\end{align*}

Utilizing the Lagrange multiplier approach, we have 
\begin{align*}
&p_i=\frac{1}{n_1} \dfrac{1}{1+ \lambda_1 (\hat{V}^{(1)}_i- \theta)}, \ \ \  i=1, ..., n_1,\\
&q_j=\frac{1}{n_2} \dfrac{1}{1+ \lambda_2 (\hat{V}^{(2)}_j- \theta)}, \ \ \  j=1, ..., n_2,
\end{align*}
where $(\lambda_1, \lambda_2, \theta)^T$ satisfies
\begin{align}\label{findtheta}
&\sum_{i=1}^{n_1} \frac{\hat{V}^{(1)}_i-\theta}{1+\lambda_1 \left[\hat{V}^{(1)}_i-\theta\right ]}=0, \nonumber \\
&\sum_{j=1}^{n_2} \frac{\hat{V}^{(2)}_j-\theta}{1+\lambda_2 \left[\hat{V}^{(2)}_j-\theta\right ]}=0, \nonumber \\
&\lambda_1 \sum_{i=1}^{n_1} \frac{-1}{1+\lambda_1 \left[\hat{V}^{(1)}_i-\theta\right ]}+\lambda_{2}\sum_{j=1}^{n_2} \frac{-1}{1+\lambda_{2} \left[\hat{V}^{(2)}_j-\theta\right ]}=0.
\end{align}
Therefore,  we can rewrite the empirical likelihood function as 
\begin{align*}
L=\prod_{i=1}^{n_1} \Big (\frac{1}{n_1} \dfrac{1}{1+ \lambda_1 (\hat{V}^{(1)}_i- \theta)} \Big ) \prod_{j=1}^{n_2} \Big(\frac{1}{n_2} \dfrac{1}{1+ \lambda_2 (\hat{V}^{(2)}_j- \theta)}\Big).
\end{align*}
Further, $\left(\prod_{i=1}^{n_1}p_i\right)  \left ( \prod_{j=1}^{n_2}q_j \right )$ is maximized at $p_i=\dfrac{1}{n_1}$ and $q_j=\dfrac{1}{n_2}$, $i=1,...,n_1, j=1,..., n_2.$
Thus, the jackknife empirical log-likelihood ratio is

\begin{align*}
l=2 \sum_{i=1}^{n_1}\log  \Big \{1+\lambda_1  \left[\hat{V}^{(1)}_i-\theta\right ] \Big  \}+
2 \sum_{i=1}^{n_2}\log \Big\{1+\lambda_{2}  \left[\hat{V}^{(2)}_i-\theta\right ] \Big \}, 
\end{align*}
where $(\lambda_1, \lambda_2, \theta)^T$ satisfies (\ref{findtheta}).

Define 
$g_1(\bi x)=\mathbb{E} h_1(\bi x, \bi X_2)-\theta_0$, $\sigma^2_{g1}=\mbox{var}\Big (g_1(\bi X_1)\Big )$, $g_2(\bi y)=\mathbb{E} h_2(\bi y, \bi Y_2)-\theta_0$, $\sigma^2_{g2}=\mbox{var}\Big (g_2(\bi Y_1)\Big )$.
We assume
\begin{itemize}
\item \textbf{C1.}  $0< \sigma_{g_k}<\infty, \ k=1, 2$;
\item \textbf{C2.} $n=n_1+n_2$, $\dfrac{n_k}{n} \to \alpha_k > 0,\ k=1,2$ and $\alpha_1+\alpha_2=1.$
\end{itemize}
Condition \textbf{C1} implies that (1) the first moment of $F$ exists and  (2) $F$ is not a point mass distribution.   Condition \textbf{C2} means that the size of one part would not dominate over the other in the partition. When $n \to \infty$, $\min (n_1,  n_2) \to \infty$.  
With those two conditions, Wilks' theorem holds. 
\begin{theorem}\label{wilk}
Under $H_0$ and the conditions \textbf{C1-C2}, we have 
$$ l \stackrel{d}{\rightarrow} \chi^{2}_{1}, \;\;\;\text{as $n \to \infty$}.$$
\end{theorem}
\begin{proof}
See the Appendix. 
\end{proof}

With Theorem \ref{wilk}, we reject $H_0$ if the observed jackknife empirical log-likelihood ratio $\hat{l} $ is greater than $\chi^2_{1, 1-\alpha}$, where $\chi^2_{1, 1-\alpha}$ is the $100(1-\alpha)\%$ quantile of $\chi^2$ distribution with $1$ degree of freedom. The $p$-value of the test can be calculated by
 $$ p\mbox{-value} = P_{H_0}(\chi^2_{1} >  \hat{l} ), $$
 and the power of the test is 
$$ \mbox{power} = P_{H_a} (l > \chi^2_{1, 1-\alpha}). $$
In the next theorem, we establish the consistence of the proposed test, which states that its power is tending to 1 as the sample size goes to infinity. 
\begin{theorem}\label{consis:test}
Under the conditions \textbf{C1-C2}, the proposed JEL test  is consistent for any fixed alternative. That is,
$$ P_{H_a} (l > \chi^2_{1, 1-\alpha}) \rightarrow 1, \;\;\;\text{as $n \to \infty$}. $$
\end{theorem}
\begin{proof}
See the Appendix. 
\end{proof}
\begin{remark}
One can test whether the distribution is symmetric about some specified center, $\bi \mu$, by letting $\bi Z=\bi X-\bi \mu$. Then testing $\bi X-\bi \mu \stackrel{D}{=} \bi \mu-\bi X$ is equivalent to testing $\bi Z \stackrel{D}{=} -\bi Z$. 
However, when the center point is unknown, the profiled JEL will be applied, which is definitely  a worthwhile subject for our future research.
\end{remark}


\section{Simulation Studies}
In order to assess the proposed JEL method for the symmetry testing, we conduct simulation studies in this section. Empirical sizes and powers for each method at significance level $\alpha=0.05$ are based on 10,000 replications. The results at level $\alpha=0.10$ are similar and hence are skipped. 
We compare the following methods. 
\begin{description}
\item[JEL:]our proposed JEL method.  R package ``dfoptim" \cite{Varadhan18} is used for solving the equation system of (\ref{findtheta}). 
\item[ET:] the DISCO test of \cite{Rizzo2010}.  The test is implemented by the permutation procedure. Function ``eqdist.etest" with the default number of replicates in  R package ``energy" is used \cite{Rizzo17}. 
\item[M:]  the symmetry test about an unknown median for univariate data based on the Bonferroni measure of skewness (\cite{Mira1999}). 
\item[CM:]  the symmetry test about an unknown median for univariate data using $C=\dfrac{\bar X-M}{s}$ as the test statistic, where $\bar X$, $M$ and $s$ are the sample mean, median and standard deviation, respectively (\cite{Cabilio1996}).
\item[MGG:] a robustified version of CM (\cite{Miao2006}).  The procedure ``lawstat" in R package is used.  
\item[CD:] the symmetric test based on the distance of characteristic functions (\cite{Chen2019}). Similar to the ET method, a permutation procedure is applied and the number of permutations is set to the same as the ET method.  
\end{description}  
Note that the methods of M, CM and MGG are only for univariate distributions.  Generating a  sample from the population and then splitting the sample into two samples will be the same as generating two independent samples from the same distribution. Therefore, in this section, we generate two independent samples with either equal  ($n_1=n_2=20$ or $n_1=n_2=50$) or unequal sample sizes ($n_1=40, n_2=60$)  from a $d$-dimensional distribution $F$. 

In  order to assess the empirical sizes of the proposed JEL, we consider the following diagonal symmetric distributions. Results are reported in Table \ref{tab:es1} for the univariate distributions  and \ref{tab:es} for the multivariate distributions.  
 \begin{itemize}
 \item Univariate case:  $N(0,1)$, $t_5$, Laplace, Logistic distributions with location of 0 and scale of 1, and the mixture of two normal distributions with different standard deviations, 1 and 2, respectively. 
 \item Multivariate cases $d=2,4,6$:
\begin{itemize}
 \item $F=N(\bi 0, \bi I_{d\times d})$ where $ \bi I_{d \times d}$ is the identity matrix. 
 \item  $F=N(\bi 0, \bi \Sigma_{d \times d})$, where $\bi \Sigma_{d \times d}$ is the matrix with the diagonal elements being 1 and the off-diagonal elements being 0.5.
\item $F = t_5(\bi 0, \bi\Sigma_{d \times d})$. 
\item $F$ is a mixture of $N(\bi 0, \bi I_{d\times d})$ and $N(\bi 0, \bi \Sigma_{d \times d})$ with the mixing proportions of 0.9 and 0.1, respectively. 
\end{itemize}
 \end{itemize}

\begin{table}[ht]
\begin{center}
\caption{Empirical sizes of each test on univariate distributions under the significance level 0.05.}\label{tab:es1}
\begin{tabular}{c|cccccccc}\hline\hline
     Distribution&Sample Size&  JEL &ET&MGG&CM&M &CD \\ \hline
     & $n_1=n_2=20$       &.058&.044&.041&.037&.031&.055 \\
   N(0,1)   & $n_1=40,n_2=60$       &.048&.044&.049&.048&.044 &.057\\
      & $n_1=n_2=50$      &.048&.044&.049&.048&.044 &.057\\
\midrule

      & $n_1=n_2=20$     &.075&.048&.057&.035&.032&.056 \\
  $ t_5(0,1)$   & $n_1=40,n_2=60$       &.060&.049&.060&.042&.043&.058 \\
      & $n_1=n_2=50$       &.057&.043&.061&.041&.043&.054 \\
\midrule
   & $n_1=n_2=20$       &.078&.048&.062&.032&.032&.056 \\
   Laplace(0,1)   & $n_1=40,n_2=60$      &.051&.048&.063&.033&.040&.057 \\
      & $n_1=n_2=50$       &.054&.048&.063&.033&.040& .057\\
\midrule
      & $n_1=n_2=20$       &.069&.047&.046&.034&.032&.057 \\
   Logistic(0,1)   & $n_1=40,n_2=60$       &.042&.048&.047&.037&.040&.058 \\
      & $n_1=n_2=50$        &.048&.048&.047&.037&.040&.058 \\
\midrule

      & $n_1=n_2=20$       &.055&.043&.047&.038&.033&.053 \\
    $90\%N(0, 1)+10\%N(0, 2)$& $n_1=40,n_2=60$      &.041&.043&.058&.047&.047 &.056\\
      & $n_1=n_2=50$      &.041&.043&.054&.047&.045&.055 \\ 
 \hline\hline
\end{tabular}
\end{center}
\end{table}

The proposed JEL is only compared with ET and CD approaches in Table \ref{tab:es} since other methods of MGG, CM and M are not valid for multivariate cases. 
From Tables \ref{tab:es1} and \ref{tab:es} , we see that in the heavy-tailed distributions, although JEL has a slight oversize problem while CM and M are under-sized, those problems are not longer an issue when the sample size increases. Overall,  the empirical sizes of all methods are fairly close to the nominal levels.  Furthermore,  the dimensions and the choices of $n_1$ and $n_2$ do not have too much impact on the control of the sizes. The proposed JEL test  can control the sizes pretty well in different dimensions and in unequal cases. 
\begin{table}[]
\begin{center}
\caption{Empirical sizes of each test on multivariate distributions under the significance level 0.05.}\label{tab:es}
\begin{tabular}{c|ccccc}\hline \hline

               Distribution &Simension&Sample Size&   JEL &ET& CD 
\\ \hline


%
   &    & $n_1=n_2=20$       &.051&.046&.053\\
    &$d=2$&$n_1=40, n_2=60$          &.044&.043&.053\\
    &&$n_1=n_2=50$          &.045 &.043&.055\\\cmidrule{2-6}

   &    & $n_1=n_2=20$       &.048&.045&.056\\
$N(\bi 0, \bi I_{d\times d})$    &$d=4$&$n_1=40, n_2=60$          &.043&.043&.056\\
    &&$n_1=n_2=50$         &.046 &.046&.054\\\cmidrule{2-6}

   &    & $n_1=n_2=20$      &.041&.045&.055\\
    &$d=6$&$n_1=40, n_2=60$        &.044&.044&.053\\
    &&$n_1=n_2=50$         &.047 &.046&.055\\

 \midrule

   &    & $n_1=n_2=20$&.053&.046&.054\\
   &$d=2$&$n_1=40, n_2=60$          &.047&.046&.054\\
    &&$n_1=n_2=50$         &.045 &.044&.055\\\cmidrule{2-6}

   &    & $n_1=n_2=20$      &.052&.042&.056\\
$N(\bi 0, \bi \Sigma_{d\times d})$     &$d=4$&$n_1=40, n_2=60$        &.043&.046&.055\\
    &&$n_1=n_2=50$         &.045 &.044&.054\\\cmidrule{2-6}

   &    & $n_1=n_2=20$      &.051&.046&.054\\
    &$d=6$&$n_1=40, n_2=60$          &.046&.043&.050\\
    &&$n_1=n_2=50$         &.044 &.046&.055\\

 \midrule

    &    & $n_1=n_2=20$      &.050&.043&.052\\
   &$d=2$&$n_1=40, n_2=60$        &.043&.046&.055\\
    &&$n_1=n_2=50$         &.046 &.041&.055\\\cmidrule{2-6}

   &    & $n_1=n_2=20$      &.046&.043&.054\\
$90\%N(\bi 0, \bi I_{d\times d})+10\%N(\bi 0, \bi \Sigma_{d\times d})$     &$d=4$&$n_1=40, n_2=60$         &.045&.043&.054\\
    &&$n_1=n_2=50$         &.045 &.040&.051\\\cmidrule{2-6}

   &    & $n_1=n_2=20$      &.045&.045&.057\\
    &$d=6$&$n_1=40, n_2=60$          &.047&.044&.051\\
    &&$n_1=n_2=50$         &.047 &.045&.056\\

 \midrule 
   &    & $n_1=n_2=20$      &.066&.043&.051\\
   &$d=2$&$n_1=40, n_2=60$        &.055&.048&.058\\
    &&$n_1=n_2=50$          &.051 &.044&.054\\\cmidrule{2-6}

  &    & $n_1=n_2=20$       &.102&.045&.058\\
$t_5(\bi 0, \bi \Sigma_{d \times d})$    &$d=4$&$n_1=40, n_2=60$         &.051&.047&.051\\
    &&$n_1=n_2=50$         &.047 &.046&.053\\\cmidrule{2-6}

   &    & $n_1=n_2=20$    &.180&.046&.054\\
    &$d=6$&$n_1=40, n_2=60$         &.068&.046&.053\\
    &&$n_1=n_2=50$         &.056 &.046&.057  \\ \hline\hline

\end{tabular}
\end{center}
\end{table}

In the second part of the simulation study, we consider the distributions not symmetric around $\bi 0$ to assess the empirical powers. 
The following distributions in $\mathbb{R} ^d$ $(d=2, 4, 6)$ are considered and results are reported in Table \ref{tab:ep}. 
 \begin{itemize}
 \item $F= N(\bi{0.5}, \bi I_{d\times d})$, where $\bi{0.5}$ is the $d$-vector with all elements being 0.5.  
 \item  $F=N(\bi{0.5}, 16 \bi \Sigma_{d \times d})$.
\item $F$ is a mixture of $N(\bi 0, \bi I_{d\times d})$ and $N(\bi {0.5}, 16\bi \Sigma_{d \times d})$ with an equal mixing proportion.   
\item $F = t_5(\bi{0.5}, 16 \bi \Sigma_{d \times d})$.  
 \end{itemize}

%

\begin{table}[]
\begin{center}
\caption{Empirical powers of each test under the significance level 0.05.}\label{tab:ep}
\begin{tabular}{c|ccccc}\hline \hline

               Distribution &Dimension&Sample Size&   
JEL &ET& CD\\ \hline



   &    & $n_1=n_2=20$       &.158&.973&.944\\
    &$d=2$&$n_1=40, n_2=60$         &.390&1.00&1.00\\
    &&$n_1=n_2=50$         &.416 &1.00&1.00\\\cmidrule{2-6}

   &    & $n_1=n_2=20$      &.274&.999&.997\\
$N(\bi{0.5}, \bi I_{d\times d})$    &$d=4$&$n_1=40, n_2=60$          &.642&1.00&1.00\\
    &&$n_1=n_2=50$        & .686&1.00&1.00\\\cmidrule{2-6}

   &    & $n_1=n_2=20$       &.402&1.00&1.00\\
    &$d=6$&$n_1=40, n_2=60$          &.803&1.00&1.00\\
    &&$n_1=n_2=50$        &.846&1.00&1.00\\

 \midrule
%
%
%
%
%
%
%
%
%


   &    & $n_1=n_2=20$       &.867&.119&.080\\
    &$d=2$&$n_1=40, n_2=60$        &.573&.248&.116\\
    &&$n_1=n_2=50$        &.535 &.262&.117\\\cmidrule{2-6}

   &    & $n_1=n_2=20$      &.993&.137&.069\\
$N(\bi{0.5},  16 \bi \Sigma_{d\times d})$    &$d=4$&$n_1=40, n_2=60$       &.933&.287&.094\\
    &&$n_1=n_2=50$          &.904 &.299&.093\\\cmidrule{2-6}

   &    & $n_1=n_2=20$      &.996&.144&.068\\
    &$d=6$&$n_1=40, n_2=60$          &.990&.310&.083\\
    &&$n_1=n_2=50$          &.979 &.322&.080\\

 \midrule


   &    & $n_1=n_2=20$       &.521&.061&.059\\
    &$d=2$&$n_1=40, n_2=60$       &.220&.099&.062\\
    &&$n_1=n_2=50$         &.207&.109&.063\\\cmidrule{2-6}

   &    & $n_1=n_2=20$       &.934&.063&.057\\
$50\%N(\bi 0, \bi I_{d \times d})+50\%N(\bi{0.5}, 16\bi \Sigma_{d\times d})$    &$d=4$&$n_1=40, n_2=60$          &.822&.120&.053\\
    &&$n_1=n_2=50$         &.804 &.120&.060\\\cmidrule{2-6}

   &    & $n_1=n_2=20$      &.987&.066&.058\\
    &$d=6$&$n_1=40, n_2=60$          &.942&.128&.055\\
    &&$n_1=n_2=50$        &.937&.132&.054\\

 \midrule


   &    & $n_1=n_2=20$       &.972&.102&.075\\
    &$d=2$&$n_1=40, n_2=60$         &.936&.204&.116\\
    &&$n_1=n_2=50$         &.932&.204&.114\\\cmidrule{2-6}

   &    & $n_1=n_2=20$       &.990&.118&.079\\
$t_5(\bi{0.5}, 16 \bi \Sigma_{d\times d})$    &$d=4$&$n_1=40, n_2=60$          &.994&.233&.101\\
   &&$n_1=n_2=50$          &.992&.239&.102\\\cmidrule{2-6}

   &    & $n_1=n_2=20$      &.988&.127&.070\\
    &$d=6$&$n_1=40, n_2=60$         &.995&.260&.097\\
    &&$n_1=n_2=50$          &.997 &.262&.096\\

%
%
%
%
%

 \hline\hline

\end{tabular}
\end{center}
\end{table}

In Table \ref{tab:ep}, we observe that JEL test has a low power in $N(\bi 0.5, \bi I_{d\times d})$, although it is consistent for any alternative. In that case, the ET and CD  method achieve higher power than the JEL test. One intuitive explanation is that the JEL approach assigns  more weights on the sample points close to $\bi 0$ and hence the JEL loses some power to reject the symmetry. The same reason leads to low powers for the JEL test in detecting location differences of two distributions \cite{Sang2019b}. Such a phenomenon is also common for the tests based on the density function approach, as mentioned in \cite{Martinez09}.  However, when the covariance is departure from the identity, especially for high dimensions, there are less sample points around $\bi 0$.  JEL is much more powerful than the ET  and CD approaches.  One can clearly find that the ET and CD approaches are dull to detect the asymmetry when the scatter matrix of the distribution is not identity, and the JEL method is efficient to detect this type of asymmetry, especially for high dimensions. It can also be seen clearly from the results that the empirical powers increase as the dimension increases from $d=2$ to $d=6$,  which is consistent with the results in \cite{Chen2019} and \cite{Szekely13}.

\section{Conclusion}

In this paper, we have developed a consistent JEL test for diagonal symmetry  based on energy statistics. The standard limiting chi-square distribution with one degree of  freedom is established and is used to conduct hypothesis testings without a permutation procedure. Numerical studies confirm the advantages of the proposed method under a variety of situations.  One of important contributions of this paper is to develop a powerful nonparametric method for diagonal symmetry test. This can also be extended to test the central symmetry around some specified centers.  
However, if the center is not provided, we need to estimate the center, which is definitely a worthwhile subject for future research.

\section{Appendix}
Define  $\bi \lambda=(\lambda_1, \lambda_{2})$,
\begin{align*}
&W_{1n}(\theta, \bi \lambda)=\frac{1}{n}\sum_{i=1}^{n_1} \frac{\hat{V}^{(1)}_i-\theta}{1+\lambda_{1} \left[\hat{V}^{(1)}_i-\theta\right ]},\\
&W_{2n}(\theta, \bi \lambda)=\frac{1}{n}\sum_{i=1}^{n_2} \frac{\hat{V}^{(2)}_i-\theta}{1+\lambda_{2} \left[\hat{V}^{(2)}_i-\theta\right ]},\\
&W_{{3n}}(\theta, \bi \lambda)=\frac{1}{n}\lambda_{1}\sum_{i=1}^{n_1} \frac{-1}{1+\lambda_{1} \left[\hat{V}^{(1)}_i-\theta\right ]}+\frac{1}{n}\lambda_{2}\sum_{i=1}^{n_2} \frac{-1}{1+\lambda_{2} \left[\hat{V}^{(2)}_i-\theta\right ]}.
\end{align*}
\begin{lemma}[Hoeffding, 1948]\label{Hoeffding} 
If $\E[h^2_1(X_1, X_2)] < \infty, \ \E[h^2_2(Y_1, Y_2)] < \infty, \ \sigma^2_{g_1}>0, \ \sigma^2_{g_2}>0$
then
 $$\dfrac{\sqrt{n_1}(U_{1}-\theta_0)}{2\sigma_{g_1}} \stackrel{d}{\rightarrow} N(0, 1) \ \text{ as}  \ n_1 \to \infty,$$
 $$\dfrac{\sqrt{n_2}(U_{2}-\theta_0)}{2\sigma_{g_2}} \stackrel{d}{\rightarrow} N(0, 1) \ \text{ as}  \ n_2 \to \infty.$$
\end{lemma}

\begin{lemma}\label{JingA.3}
Let $S_{j}=\dfrac{1}{n_j} \sum_{i=1}^{n} \left [\hat{V}^{(j)}_i- \theta_0\right ]^2,\  j=1, 2.$
Under the conditions of Lemma \ref{Hoeffding}, 
\begin{align*}
S_{j}=4\sigma^2_{g_j}+o_P(1),  \ n_j \to \infty,\ j=1,2.
\end{align*}
\end{lemma} 
\begin{proof}
See Lemma A.3 in \cite{Jing2009}.
\end{proof}

\begin{lemma}[Liu, Liu and Zhou, 2018]\label{Liu5.5}
Under the conditions of Lemma \ref{Hoeffding}, and under $H_0$, with probability tending to one as $n \to \infty$, there exists a root $\tilde{\theta}$ of 
\begin{align*}
&W_{jn}(\theta, \bi \lambda)=0, \ j=1, 2, 3,
\end{align*}
such that $|\tilde{\theta}-\theta_0|< \delta,$ where $\delta=n^{-1/3}$. 
\end{lemma}

Let $\tilde{\eta}=(\tilde{\theta}, \tilde{\bi \lambda})^T$ be the solution to the above equations, and $\eta_0=(\theta_0, 0, 0)^T$. By expanding $W_{jn}(\tilde{\eta})$ at $\eta_0$, we have, for $j=1, 2, 3$,
\begin{align}\label{W_n}
0=W_{jn}(\eta_0)+ \dfrac{\partial W_{jn}}{\partial \theta}(\eta_0)(\tilde{\theta}-\theta_0)+\dfrac{\partial W_{jn}}{\partial \lambda_1}(\eta_0) \tilde{\lambda}_1+\dfrac{\partial W_{jn}}{\partial \lambda_2}(\eta_0) \tilde{\lambda}_{2}+R_{jn},
\end{align}
where $R_{jn}=\dfrac{1}{2} (\tilde{\eta}-\eta_0)^T \dfrac{\partial^2 W_{jn}(\eta^{*})}{\partial \eta \partial \eta^T} (\tilde{\eta}-\eta_0)$ where $\eta^{*}$ lies between $\eta_0$ and $\tilde{\eta}$. One can easily prove that 
$R_{jn}=o_P(n^{-1/2}).$

\begin{lemma}\label{cov:u}
$\mbox{Cov}(U_{1}, U_{2})=0$.
\end{lemma}

\textbf{Proof of Theorem \ref{wilk}}

By Lemma \ref{JingA.3} and equation (\ref{W_n}),
$$
\begin{pmatrix}
W_{1n}(\eta_0))\\
   W_{2n}(\eta_0)   \\
   0\\
\end{pmatrix}
= W
\begin{pmatrix}
    \tilde{\lambda}_{1}\\
\tilde{\lambda}_{2} \\
   \tilde{\theta}-\theta_0
    \end{pmatrix}+o_{P}(n^{-1/2}),
$$
where 
$$ W=
\begin{pmatrix}
 \alpha_1 \sigma_1^2&0&\alpha_1 \\
 0&\alpha_2 \sigma_2^2&\alpha_2\\
      \alpha_1&\alpha_2&0\\   
\end{pmatrix},$$
$\sigma^2_j=4\sigma^2_{g_j}, j=1, 2.$ 

 $W$ is nonsingular. Therefore,
$$\begin{pmatrix}
    \tilde{\lambda}_{1}\\
\tilde{\lambda}_{2} \\
   \tilde{\theta}-\theta_0
    \end{pmatrix}=W^{-1}\begin{pmatrix}
   W_{1n}(\eta_0)   \\
  W_{2n}(\eta_0)  \\
  0\\
\end{pmatrix}+o_{P}(n^{-1/2}).
$$

Under $H_0$, $\sigma^2_1=\sigma^2_2:=\sigma^2$.
\begin{align*}
&\tilde{\theta}-\theta_0=W_{1n}(\eta_0)+W_{2n}(\eta_0)+o_p(n^{-1/2}).
\end{align*}

By \cite{Jing2009}, we have 
\begin{align*}
&\tilde{\lambda}_{j}=\dfrac{U_{j}-\tilde{\theta}}{\tilde{S}_{j}}+o_{P}(n^{-1/2}), \ j=1, 2,\\
\end{align*}
where 
$\tilde{S}_{j}=\dfrac{1}{n_j}\sum^{n_j}_{i=1}(\hat{V}^{(j)}_{i}-\tilde{\theta})^2$. It can be easily to check that $\tilde{S}_{j}=\sigma^2_j+o_p(1),\ j=1,2$.
By the proof of Theorem 1 in \cite{Jing2009},
\begin{align*}
l(\theta_0)=\left [ \sum_{j=1}^2 n_j \dfrac{(U_{j}-\tilde{\theta})^2}{\sigma^2_{j}} \right](1+o_{P}(1)).
\end{align*}
By simple algebra,
\begin{align}\label{decomposition}
\sum_{j=1}^2 n_j \dfrac{(U_{j}-\tilde{\theta})^2}{\sigma^2_{j}} =(\sqrt{n}W_{1n}(\eta_0), \sqrt{n}W_{2n}(\eta_0),) \times \bi A^T \bi D \bi A
\times(\sqrt{n}W_{1n}(\eta_0), \sqrt{n}W_{2n}(\eta_0))^T +o_p(1),
\end{align}
where
$$ \bi D=
\begin{pmatrix}
\alpha_1/\sigma^2_1&0\\
 0&\alpha_2/\sigma^2_2\\
\end{pmatrix}$$
and 
$$ \bi A=
\begin{pmatrix}
1/\alpha_1-1&-1\\
 -1&1/\alpha_2-1\\
\end{pmatrix}.$$

By Hoeffding decomposition and Lemma \ref{cov:u},
$$
\sqrt{n}\begin{pmatrix}
   W_{1n}(\eta_0)   \\
  W_{2n}(\eta_0)  \\
\end{pmatrix}
\to 
N(\bi 0, \bi \Sigma),
$$
where 
$$ \bi \Sigma=
\begin{pmatrix}
\alpha_1 \sigma^2_{1}&0\\
0&\alpha_2 \sigma^2_{2}\\
\end{pmatrix}.$$

Hence, under $H_0$, the empirical log likelihood ratio converges in distribution to $\sum_{i=1}^2 \omega_i \chi^2_i$, where $\chi^2_i, i=1,2$ are two independent chi-square random variables with one degree of freedom, and $\omega_i, i=1, 2$ are eigenvalues of  $\bi \Sigma^{1/2}_0 \bi A^T \bi D_0  \bi A \bi \Sigma^{1/2}$, where 

$$ \bi \Sigma_0=\sigma^2
\begin{pmatrix}
\alpha_1&0\\
0&\alpha_2\\
\end{pmatrix}$$  and  
$$ \bi D_0=\dfrac{1}{\sigma^2}
\begin{pmatrix}
\alpha_1&0\\
 0&\alpha_2\\
\end{pmatrix}.$$
After a simple algebra calculation, $\bi \Sigma_0^{1/2} \bi A^T \bi D_0 \bi A  \bi \Sigma_0^{1/2}$ has two eigenvalues 0, 1.  $\hfill{\Box}$ 


\medskip
\textbf{Proof of Theorem \ref{consis:test}}
Under $H_a$, $\mathbb{E}U_k,$ $k=1, 2$ will be different. Let $\mathbb{E}U_k=\theta_k, \ k=1, 2.$
From (\ref{decomposition}), 
\begin{eqnarray*}
 l&=&\sum^2_{k=1}\dfrac{n_k(U_k-\tilde{\theta})^2}{\tilde{S}^2_i}+o(1)\\
&=& \sum^2_{k=1}\left [\dfrac{\sqrt{n_k}(U_k-\theta_k)}{\tilde{S}_k}+\dfrac{\sqrt{n_k}(\theta_k-\tilde{\theta})}{\tilde{S}_k} \right ]^2+o(1),\\
\end{eqnarray*}
which will be divergent since at least one of $\dfrac{\sqrt{n_k}(\theta_k-\tilde{\theta})^2}{\tilde{S}_k}, k=1, 2$ will diverge to $\infty$. 
$\hfill{\Box}$


\end{document}